\documentclass[runningheads,envcountsame]{llncs}

\usepackage[T1]{fontenc}
\usepackage{microtype}
\usepackage{graphicx} 
\usepackage{todonotes}
\usepackage{amsfonts} 
\usepackage{amssymb} 
\usepackage{amsmath} 
\usepackage{xspace} 
\usepackage{mathtools} 
\usepackage{booktabs}
\usepackage{hyperref}
\usepackage{tabularx}
\usepackage{cleveref}
\usepackage{orcidlink}

\title{Logic and Computation through the Lens of Semirings}

\author{Timon Barlag\inst{1}\orcidlink{0000-0001-6139-5219} \and
Nicolas Fröhlich\inst{1}\orcidlink{0000-0001-6139-5219} \and
Teemu Hankala\inst{2}\orcidlink{0009-0009-5535-5718} \and 
Miika Hannula\inst{2,3}\orcidlink{0000-0002-9637-6664} \and
Minna Hirvonen\inst{2}\orcidlink{0000-0002-2701-9620} \and
Vivian Holzapfel\inst{1}\orcidlink{0009-0001-1439-5037} \and
Juha Kontinen\inst{2}\orcidlink{0000-0003-0115-5154} \and
Arne Meier\inst{1}\orcidlink{0000-0002-8061-5376} \and
Laura Strieker\inst{1}\orcidlink{0009-0005-4878-4953}} 

\institute{Leibniz University Hannover, Germany\\
\email{\{barlag,nicolas.froehlich,holzapfel,meier,strieker\}@thi.uni-hannover.de}\\
\and
University of Helsinki, Finland\\
\email{\{teemu.hankala,minna.hirvonen,juha.kontinen\}@helsinki.fi}
\and
University of Tartu, Estonia\\
\email{miika.hannula@ut.ee}}
\date{May 2025}
\authorrunning{T. Barlag et al.}

\newenvironment{restate}[1][\unskip]{
    \begingroup
    \def\thetheorem{#1}
}
{
    \addtocounter{theorem}{-1}
    \endgroup
}

\newcommand{\arb}{\mathrm{Arb}_K}
\newcommand{\arbF}{\mathrm{Arb}_K}
\newcommand{\enc}{\textit{enc}}
\newcommand{\lit}{\textit{Lit}}
\newcommand{\logicFont}[1]{\protect\ensuremath{\mathrm{#1}}\xspace}
\newcommand{\FOK}{\ensuremath{{\logicFont{FO}}_{K}}\xspace}
\newcommand{\FO}{\ensuremath{{\logicFont{FO}}}\xspace}

\newcommand{\R}{\mathbb{R}}

\newcommand{\N}{\mathbb{N}}
\newcommand{\Z}{\mathbb{Z}}
\newcommand{\ol}{\overline}
\newcommand{\bO}{\mathcal{O}}
\newcommand{\ar}{\ensuremath{\textnormal{ar}}\xspace}

\newcommand{\ddfn}{\Coloneqq}

\newcommand{\BSSK}{\mathrm{BSS}_K}
\newcommand{\PK}{\mathrm{P}_K}

\newcommand{\calC}{\mathcal{C}}

\newcommand{\FACO}{\mathrm{FAC}^0_{K}}
\newcommand{\size}{{\mathrm{size}}}
\newcommand{\depth}{{\mathrm{depth}}}

\usepackage{booktabs}
\newcommand{\decproblemdef}[3]{%
\begin{center}
\begin{tabular}{lp{10cm}}\toprule
\textsf{\bfseries Problem:}& #1 \\\midrule
\textsf{\bfseries Input:}& #2\\
\textsf{\bfseries Question:}& #3?\\\bottomrule
\end{tabular}
\end{center}
}

\newcommand{\fv}{\mathrm{FV}}
\newcommand{\Var}{\mathrm{Var}}

\newcommand{\funcproblemdef}[3]{%
\begin{center}
\begin{tabular}{lp{10cm}}\toprule
\textsf{\bfseries Problem:}& #1 \\\midrule
\textsf{\bfseries Input:}& #2\\
\textsf{\bfseries Output:}& #3\\\bottomrule
\end{tabular}
\end{center}
}

\newcommand{\EVAL}{\text-\mathrm{EVAL}}
\newcommand{\MC}{\text-\mathrm{MC}}
\newcommand{\FOKMC}[2][O]{\FOK(#1)\MC\def\temp{#2}\ifx\temp\empty\else_#2\fi}

\usepackage{stmaryrd}
\newcommand{\evaluate}[2]{\protect\ensuremath{\llbracket#1\rrbracket_{#2}}}

\usepackage{tikz}
\usetikzlibrary{
    arrows,
    positioning,
}

\usepackage[linesnumbered, ruled, vlined]{algorithm2e}
\newcommand{\procEval}{{\normalfont\texttt{Eval}}}

\newcommand{\PSPACEK}{\mathsf{PSPACE}_K}

\newcommand{\ourpar}[1]{%
  \par\vspace{\baselineskip}%
  \noindent\textbf{\sffamily #1. }%
}

\begin{document}

\maketitle

\begin{abstract}
We study the expressivity and computational aspects of first-order logic and its extensions in the semiring semantics developed by Gr\"adel and Tannen. 
We characterize the complexity of model checking and data complexity of first-order logic both in terms of a generalization of Blum-Shub-Smale machines and arithmetic circuits defined over a semiring. 
In particular, we give a logical characterization of constant-depth arithmetic circuits by an extension of first-order logic that holds for any semiring that is both commutative and positive. 
\end{abstract}

\section{Introduction}
In the last decade, the use of semirings to study provenance for database query languages has attracted increasing amounts of attention~\cite{gradel17semiringprovenancefirstordermodel,DannertGNT21,DBLP:conf/lics/GradelHNW22,abs-1712-01980,GreenKT07}.
In this article, we study the expressivity and computational aspects of first-order logic in the semiring semantics originating from these studies of provenance in databases.

Semirings are algebraic structures that generalize rings by relaxing the requirement for additive inverses. 
They have found numerous applications in computer science and AI due to their versatility and modularity in modelling and analyzing computational problems~\cite{DBLP:journals/fuin/RudeanuV04,DBLP:journals/tcs/Peeva91,DBLP:journals/ijac/GaubertK06,DBLP:journals/soco/LitvinovRSS13,DBLP:journals/cj/Moller13,DBLP:journals/ijfcs/EsparzaLS15,DBLP:journals/order/Vrana22,EiterK23}.
In other words, specific semirings correspond to different computational paradigms or problem domains. 
Most of the classical complexity theory lives in the domain of the Boolean semiring $\mathbb{B}$, whereas the semiring of natural numbers $\mathbb{N}$ is the domain of problems related to counting, and the probability semiring $\mathbb{R}_{\geq 0}$ for problems with geometric or continuous features. 
On the other hand, so-called tropical semirings have various applications in performance analysis~\cite{DBLP:journals/access/OmanovicOC23} and reachability problems~\cite{DBLP:journals/ijac/GaubertK06}.

Several computation models can be generalized to various classes of semirings. 
For example, weighted automata and weighted Turing machines label the transitions of the machine with semiring elements representing quantities such as probabilities, costs, or capacities. 
Different semirings enable automata to model varied quantitative behaviours with a wide array of applications (see, e.g., \cite{KOSTO,badia}). 
Furthermore, circuit complexity theory and algebraic algorithms readily generalize to various families of semirings (see \cite{ganardi} and the references therein).    

Semirings have also found applications in databases. Semiring provenance is an approach to query evaluation in which the result of a query is something more than just a mere one-bit true/false answer. 
The basic idea behind this approach is to annotate the atomic facts in a database by values from some semiring $K$, and to propagate these values through a query.
Depending on the choice of the semiring, the provenance valuation gives information about a query, e.g., regarding its confidence, cost, or the number of assignments that make the query true \cite{GreenKT07}.
Semiring semantics for query languages is currently an actively studied topic in database theory (see, e.g.,
 \cite{10.1145/3651146,im_et_al:LIPIcs.ICDT.2024.11} for Datalog queries and  \cite{eldar_et_al:LIPIcs.ICDT.2024.4,munozserrano_et_al:LIPIcs.ICDT.2024.12,tencate_et_al:LIPIcs.ICDT.2024.8} for conjunctive queries). 
  
Semiring semantics has also been defined for first-order logic \cite{abs-1712-01980,Tannen17}. In this context, it has particularly been explored via key themes in classical finite model theory, including Ehrenfeucht--Fra\"{i}ss\'{e} games, locality, 0-1 laws, and definability up to isomorphisms \cite{BiziereGN23,BrinkeGM24,GradelHNW22,GradelM21}.
Recently, semiring semantics has been further extended to more expressive logical languages such as fixed-point logic \cite{DannertGNT21} and team-based logics \cite{BarlagHKPV23}.  
It is worth noting that the logics in these works differ from the logics that have been defined and studied in the context of weighted automata and Turing machines \cite{KOSTO,badia}. 
In fact, a natural
computation model for our purposes is not a weighted TM but a generalization of the Blum-Shub-Smale machine (BSS machine) that we define in this article. The inputs of such a machine are finite sequences of the elements of the semiring $K$ whereas the weighted machines operate with classical Boolean inputs.
\ourpar{Contributions}
In this paper,  we take the first steps towards connecting logics in the semiring semantics with computations defined over semirings. As mentioned above, the extensively studied framework of weighted machines and logics is not applicable to our context.  
We consider novel extensions of first-order logic by so-called formula inequalities and built-in $K$-relations in the semiring semantics. The former extension turns out to be an essential feature of the logics to be able to characterize non-trivial computations.
In order to characterize the complexity of these logics, we generalize the well-known BSS machines~\cite{BSSbook} to arbitrary semirings. We also define another polynomially equivalent machine model, the $K$-Turing machine, which facilitates the proofs and may be of independent interest.  
In the second part of the article, we consider arithmetic circuits over a semiring $K$ and give a logical characterisation of the function class (non-uniform) $\FACO$ by an extension of first-order logic that is true for any semiring~$K$ that is commutative and positive. 

\ourpar{Organisation}
In Section~\ref{BSS}, we generalize the BSS model from the reals to a wide variety of  semirings, define the $K$-Turing machine model,  and show how classical computations can be simulated on such  machines. 
In Section~\ref{Circuits}, we go through the basic definitions regarding arithmetic circuits over a semiring. 
In Section~\ref{logic}, we define the semiring interpretation of first-order formulas. In Section~\ref{results}, 
we characterize the complexity of model checking and the data complexity of first-order logic over a semiring $K$ in terms of a generalization of BSS machines and in Section~\ref{circ} we give a characterization via arithmetic circuits defined over~$K$.


\section{Preliminaries}
We assume familiarity with basic concepts in theoretical computer science, e.g., Turing machines~\cite{DBLP:books/daglib/0086373}. 
We start with the fundamental definition of a semiring.
\begin{definition}
    A \emph{semiring} is a tuple $K=(K,+,\cdot,0,1)$, where $+$ and $\cdot$ are binary operations on $K$, $(K,+,0)$ is a commutative monoid with identity element $0$, $(K,\cdot ,1)$ is a monoid with identity element $1$, $\cdot$ left and right distributes over $+$, and $x \cdot 0 =0= 0\cdot x$ for all $x \in K$.
    $K$ is called \emph{commutative} if $(K,\cdot ,1)$ is a commutative monoid.
    As usual, we often write $ab$ instead of $a \cdot b$. 
    A semiring is \emph{ordered} if there exists a  
    partial order $\leq$ (meaning $\leq$ is reflexive, transitive, and antisymmetric) such that for all $a,b, c \in K$ $a\leq b \implies a+c \leq b+c$ and $a\leq b \land 0 \leq c \implies ac \leq bc \land ca \leq cb$.
    A semiring is \emph{positive} if it has no divisors of $0$, i.\,e., $ab\neq 0$ for all $a,b\in K$, where $a\neq 0 \neq b$ and if $a+b=0$ implies that $a=b=0$.
\end{definition}
The \emph{natural order} $\leq_{\rm n}$ of a semiring $K$ is  defined by $a\leq_{\rm n} b \iff \exists c: a+c=b$. 
 The natural order of a semiring is always a preorder (meaning $\leq$ is reflexive and transitive). Moreover, if $\leq_{\rm n}$ is antisymmetric, then the semiring is ordered by it.

Important examples of semirings are the \emph{Boolean semiring} $\mathbb{B}=(\mathbb{B},\lor,\land,0,1)$ as the simplest example of a semiring that is not a ring, the \emph{probability semiring} $\mathbb{R}_{\geq 0}=(\mathbb{R}_{\geq 0},+,\cdot,0,1)$,
and the \emph{semiring of natural numbers} $\mathbb{N}=(\mathbb{N},+,\cdot,0,1)$. 
Further examples include
the \emph{tropical semiring} $\mathbb{T} = (\mathbb{R}\cup\{\infty\}, \min, +, \infty, 0)$,
the \emph{semiring of multivariate polynomials} $\mathbb{N}[X]=(\mathbb{N}[X],+,\cdot,0,1)$,
and the \emph{\L{}ukasiewicz semiring} $\mathbb{L} = ([0,1], \max, \cdot , 0, 1)$, used in multivalued logic.

\begin{remark}
Throughout this paper, we consider semirings that are commutative, positive, and nontrivial (i.e., $0 \neq 1$). These assumptions are generally assumed for simplicity, even though they are not universally needed throughout the paper.
\end{remark}


\subsection{\texorpdfstring{BSS$_K$}{BSSK} Machines}\label{BSS}
BSS machines were introduced as a model for real computation by Blum, Shub and Smale \cite{blum1989} and were originally defined over rings.
The computation nodes of a BSS machine are usually formulated in terms of quotients of polynomial functions with real coefficients. Since semirings generally lack additive and multiplicative inverses, $\BSSK$ machines have no operations
for subtraction and division.

\begin{definition}
    For a set $A$, we let $A^* = \bigcup_{n\geq 0} A^n$ and $A^+ = \bigcup_{n >0} A^n$.
    By $A_*$ we denote the set $A^\Z$ of all $x= (\dots,x_{-2}, x_{-1}, x_0 \textbf{.} x_1, x_2, \dots)$, where \textbf{.} marks the first position.
    We define two shift operations on $A_*$: shift left $\sigma_l$, where $\sigma_l(x)_i = x_{i+1}$, and its inverse, shift right $\sigma_r$, where $\sigma_r(x)_i = x_{i-1}$, for all $i \in \Z$.
\end{definition}

\begin{definition}[$\BSSK$ machines] \label{def:shorterBSSK}
    Let $K$ be a semiring.    
    A $\BSSK$ machine $M$
    computes a (possibly partial) \emph{input-output} function $f_M \colon K^*\to K^*$ represented by a directed graph of nodes labelled by $1, \ldots, N$. These nodes are of five different types.
    A configuration of the machine $M$ at any moment of computation consists of a node $m\in \{1, \ldots ,N\}$ and a current state $x \in K_*$, proceeding as follows.
    
        The node labelled by $1$ is the unique \emph{input node} and is associated with a next node $\beta(1)$.
        Here, the input $(x_1, \ldots ,x_n) \in K^*$ of the machine is initialized as
        \[
            (\ldots , 0, \underbrace{1, \dots, 1}_{n}, 0\textbf{.} x_1, \ldots ,x_n, 0,\ldots )\in K_*.
        \]
        
        Each \emph{computation node} $m$ is associated with a next node $\beta(m)$ and a mapping $g_m\colon K_* \to K_*$ such that $g_m(x)_l=x_l$, for $l \neq i$,
        and $g_m(x)_i$ is one of $x_j+x_k$, $x_j\cdot x_k$, or $c$, for some $c \in K$ and $i,j,k\in \mathbb{Z}$.
        
        Every \emph{branch node} $m$ is associated with nodes $\beta^-(m)$ and $\beta^+(m)$.
        Given $x\in K_*$ the next node is $\beta^-(m)$ if $x_1 = x_2$ and $\beta^+(m)$, otherwise.
        If $K$ is ordered, instead the next node is $\beta^-(m)$ if $x_1 \leq x_2$ and $\beta^+(m)$, otherwise.
        
        A \emph{shift node} $m$ is associated either with $\sigma_l$ or $\sigma_r$, and a next node $\beta(m)$.
        
        The node labelled by $N$ is the unique \emph{output node}.
        Once this node is reached, the computation halts resulting in the string consisting of the first $l$ positive coordinates, where $l$ is the number of consecutive $1$s stored in the negative coordinates starting from the first negative coordinate.
    %
\end{definition}

\begin{definition}\label{def:timespace}
    A function $f \colon K^*\to K^*$ is \emph{computable} if $f=f_M$ for some machine $M$. A language $L \subseteq K^*$ is \emph{decided} by a $\BSSK$ machine $M$ if its characteristic function $\chi_L$ is $f_M$.
    A machine $M$ \emph{runs in time} $t\colon \N \rightarrow \N$ (\emph{requires space} $s\colon \N \rightarrow \N$), if $M$ reaches the output in $t(|x|)$ steps (using $s(|x|)$ coordinates of the state space) for each input $x$.
    The machine $M$ runs in \emph{polynomial time} (requires \emph{polynomial space}) if $t$ ($s$) is a polynomial function.
    The complexity class $\PK$ ($\PSPACEK$) is defined as the set of all subsets of $K^*$ that are decided by some machine $M$ running in polynomial time (requiring polynomial space).
\end{definition}

Next we define a hybrid of a $\BSSK$ machine and a Turing machine that is technically more convenient for the upcoming proofs.
In addition to elements of the semiring, the machine operates on a finite alphabet
of tape symbols. Instead of filling the tape with $0$s, a blank symbol is used. In this manner,
the lengths of the input and output strings need not be separately encoded.
The machine is also equipped with a finite number of internal registers for elements of the semiring.

\begin{definition}[$K$-TMs]\label{def:KTM}
Let $K$ be a semiring. A \emph{$K$-Turing machine}
is a tuple $M = (Q, q_0, R, P, \Gamma, b, \Sigma, \delta)$ with the following interpretations.
\begin{itemize}
\item $Q$ is a finite set of states, with $q_0 \in Q$ being the initial state.
\item $R$ is a finite set of $K$-valued registers, which are always initialized to $0$.
\item $P \colon Q \to \{=, \leq\} \times R$ is a branch predicate function.
\item $\Gamma$ is a finite set of tape symbols, with a blank symbol $b \in \Gamma$,
so that $\Gamma \cap K = \varnothing$.
\item $\Sigma \subseteq (\Gamma \setminus \{b\}) \cup K$ is the set of input symbols.
\item The type of each transition of the machine $M$ is given by a partial function
$
\delta \colon Q \times (\Gamma \cup \{\top, \bot\})
\to Q \times W
\times \mathcal{P}(R)
\times \{\sigma_l, \operatorname{id}, \sigma_r\}
$,
where
$W$ is the set $\Gamma \cup K \cup \{\operatorname{id}\} \cup (\{+, \cdot\} \times R)$.
Undefined cases represent halting conditions.
\end{itemize}
A configuration of the machine $M$ at any moment of computation consists of a state $q\in Q$
and the contents of the tape in the form $x \in (\Gamma \cup K)_*$.
An input $(x_1, \dots, x_n) \in \Sigma^*$ is initialized as
$(\ldots , b\textbf{.}\, x_1, \ldots ,x_n, b,\ldots) \in (\Gamma \cup K)_*$.

The transitions of $M$ are defined in the following manner.
Let $q$ be the current state and let $a$ be the current tape symbol.
If $a \in \Gamma$, the type of the transition is given by $\delta(q, a)$. Otherwise, $a\in K$,
and $\delta$ is applied to the pair $(q, v)$
with $v \in \{\top, \bot\}$ being the truth value of the comparison
$a \circ y$, where $P(q) = (\circ, r)$ and $y$ is the current value in the register $r$.
If $\circ$ is $\leq$ and $K$ is not ordered, the machine halts.
Then, the value $(q', z, s, \sigma)$ given by $\delta$ is applied as follows.

The next state of the machine is $q'$.
The element $z$ determines the symbol to be written on the tape under the head of the machine
so that, if $z \in \Gamma \cup K$,
the symbol $z$ is written, and if $z = \operatorname{id}$, the symbol $a$ is not changed,
and if both $z = (\star, r)$ and $a \in K$ are true, then the value of the operation
$a \star v$ is written on the tape, where $v$ is the current value of the register $r$.
In other cases $M$ halts.

If $a \in K$, the values of the registers in the set $s \subseteq R$ are set to $a$,
and in the case that both $a \notin K$ and $s \neq \varnothing$ are true, the machine halts.
The tape is shifted according to the shift operation $\sigma \in \{\sigma_l, \sigma_r\}$
or kept in place if $\sigma = \operatorname{id}$.

Once the machine halts, the result of the computation is the string consisting of the first
positive coordinates until the first blank symbol with a positive coordinate.
\end{definition}

We define the use of time and space for a $K$-Turing machine as in Definition~\ref{def:timespace}.
The computation of a $K$-TM can be simulated using a $\BSSK$ machine. This result is proved in Appendix~\ref{proof:turing}.
Simulating a $\BSSK$ using a $K$-TM is straightforward.

\begin{theorem}\label{lemma:turing}
Let $K$ be a semiring and
let $f \colon K^* \to K^*$ be a function
computed by a $K$-Turing machine $M$
that runs in time $t$ and space $s$.
Then there is a $\BSSK$ machine $M'$ computing $f$ and
$c \in \N$ such that
for each $x \in K^*$, the machine $M'$ runs on input $x$ in time $c(t(|x|) + |x|^2 + |f(x)|^2 + 1)$
and in space $c(s(|x|) + 1)$.
\end{theorem}

%

\subsection{Arithmetic Circuits}\label{Circuits} 

In
this
section we will briefly introduce relevant concepts for arithmetic circuits over semirings.
A more rigorous treatment thereof is presented in Appendix~\ref{proof:thm_fo_aco}.

\begin{definition}
    \label{def_arith_circ}
    Let $K$ be a 
    semiring.
    An \emph{arithmetic circuit} $C$ over $K$ is a connected, directed acyclic graph. 
    Its nodes (also called \emph{gates}) are each labelled as either \emph{input}, \emph{$+$}, \emph{$\cdot$}, \emph{constant} $c$ for $c \in K$,  or \emph{output}, where input and constant gates have fan-in $0$ and output gates have fan-in $1$ and fan-out $0$.
\end{definition}
The function $f_C$ computed by $C$ is defined in the obvious way.

We call the number of gates in a circuit $C$ the \emph{size} of $C$ and the longest path from an input gate to an output gate the \emph{depth} of $C$.
We also write $\size(C)$ and $\depth(C)$ to denote the respective values.

A single circuit can only compute a function with a fixed number of arguments, which is why we call arithmetic circuits a \emph{non-uniform} model of computation. 
In order to talk about arbitrary functions, we need to consider circuit families, i.e., sequences of circuits which contain one circuit for every input length $n \in \N$. 

\begin{definition}
    Let $K$ be a 
    semiring.
    A $K$-\emph{circuit family} is a sequence of $K$-circuits $\mathcal{C} = (C_n)_{n \in \N}$.
    The function computed by a circuit family $\mathcal{C} = (C_n)_{n \in \N}$ is the function computed by the respective circuit, i.e., for all $x \in K^+$
    \begin{equation}
        f_\mathcal{C}(x) = f_{C_{\lvert x\rvert}}(x).
    \end{equation} 
\end{definition}

For a function $f \colon \N \to \N$, we say that a circuit family $\mathcal{C}$ is of size $f$ (depth $f$) if the size (depth) of $C_n$ is bounded by $f(n)$ for all $n \in \N$.
We also write $\size(\mathcal{C})$ and $\depth(\mathcal{C})$ to denote the respective values.

To properly consider a circuit family $\mathcal{C}$ as an algorithm, some bounds for the computational complexity of obtaining the $n$th circuit when given $n$ are required.
This would make them a \emph{uniform} model of computation. However, questions of uniformity lie outside the scope of this paper. Accordingly, we restrict our attention to \emph{non-uniform} circuit families, though we briefly return to the issue in the conclusion.



If we want the circuits to mimic branching behaviour akin to $\BSSK$ machines, we need additional control structures.
For this reason, we add binary relations to branch on in the following way.

\begin{definition}
    Let $K$ be a 
    semiring, let $R_1,\dots ,R_l \subseteq K^2$ be relations and let $\chi_{R_i} \colon K^2 \to \{0, 1\}$ be the characteristic function of each $R_i$.
    Then the class $\FACO[R_1,\dots,R_l]$ consists of all functions $f \colon K^* \to K$ computable by $K$-circuit families $\calC$ such that $\size(\calC) \in \bO(n^{\bO(1)})$ and $\depth(\calC) \in \bO(1)$, where the circuits in $\calC$ have an additional binary gate types that compute the functions $\chi_{R_i}$.
\end{definition}
If we do not have any relations $R_i$, we write $\FACO$. 
The circuits of this class then compute only polynomials.

Of particular interest for us will be the class $\FACO[O]$ for $O \subseteq \{=, \neq, \leq, \not\leq\}$.
(Whenever we extend a logic or a computational model with an $O$, where $\leq\ \in O$ or $\not\leq\ \in O$, we assume the underlying semiring to be ordered.)
In order to characterize these classes logically, we mention another useful property thereof:
any function in $\FACO$ or the aforementioned extensions can be computed by circuits that are essentially trees.
The following lemma is a minor modification of a result about circuit over the reals~\cite[Lemma~26]{10.1093/logcom/exae051}.

\begin{lemma}\label{lem:circ_trees}
    Let $O \subseteq \{=, \neq, \leq, \not\leq\}$, let $K$ be a 
    semiring and let $f \colon K^* \to K \in \FACO[O]$.
    Then there exists a family of circuits $\calC = (C_n)_{n \in \N}$ of size $\bO(n^{\bO(1)})$ and depth $\bO(1)$ computing $f$ such that for each circuit $C \in \calC$, all non-input gates have fan-out $1$ and for each gate $g$ in $C$, each input-$g$ path has the same length.
\end{lemma}

\subsection{First-Order Logic under Semiring Semantics}\label{logic} 

In this section, 
we define the semiring semantics for  first-order formulas 
and so-called formula inequality atoms (cf. \cite{gradel17semiringprovenancefirstordermodel,BarlagHKPV23}). 
%

A vocabulary $\tau$ is a set of relation symbols. 
We denote by $\ar(R)$ the \textit{arity} of a relation  symbol $R$.
The set $\Var$ is the set of first-order variables.  



\begin{definition}[First-order logic with (in)equality]
    The syntax for the first-order logic (of a  vocabulary $\tau$) with formula (in)equality 
    for a set $O\subseteq\{{=,}{\neq,}{\leq,}{\not\leq}\}$, denoted by $\FO(O)$, is defined as follows:
     \[
    \phi\ddfn x = y \mid  x \neq y \mid R(\bar{x})\mid \neg R(\bar{x}) \mid \phi \circ \phi \mid(\phi\wedge\phi) \mid (\phi\lor\phi) \mid \exists x\phi \mid \forall x\phi,
    \]
     where $x,y\in\Var$, $R\in\tau$, $\bar{x}$ is a tuple of variables so that $\ar(R)=|\bar{x}|$, and $\circ\in O$. 
\end{definition}
If $O=\emptyset$, we write $\FO$ instead of $\FO(\emptyset)$. 
The definition of a set of \emph{free variables} $\fv(\phi)$ of a formula $\phi \in \FO(O)$ extends from $\FO$ in the obvious way: $\fv(\phi\circ \psi)=\fv(\phi)\cup\fv(\psi)$, for $\circ\in O$. If $\phi$ contains no free variables, i.e., $\fv(\phi)=\emptyset$, it is called a \emph{sentence}.

Let $A$ be a finite set (i.e., the model domain), $\tau=\{R_1,\dots,R_n\}$ a finite vocabulary, and $K$ a semiring. Define $Lit_{A,\tau}$ as the set of \textit{literals} over $A$, i.e., the set of facts and negated facts $Lit_{A,\tau}=\{R(\bar{a})\mid \bar{a}\in A^{\ar(R)},R\in\tau\}\cup\{\neg R(\bar{a})\mid \bar{a}\in A^{\ar(R)},R\in\tau\}$. A $K$-interpretation is a mapping $\pi\colon Lit_{A,\tau}\to K$. An assignment $s$ is a mapping $s\colon\Var\to A$ that associates each first-order variable with an element from the set $A$.




     
\begin{definition}\label{interpretation}
    Let $K=(K,+,\cdot,0,1)$ be a 
      semiring, and $\pi\colon Lit_{A,\tau}\to K$ a $K$-interpretation. Let $O\subseteq \{=,\neq,\leq,\not\leq\}$. If $\leq\ \in O$ or $\not\leq\ \in O$, we additionally assume that $K$ is ordered. The $K$-interpretation of a $\FO(O)$-formula $\theta$ under $\pi$ and an assignment $s$  is denoted by $\evaluate{\theta}{\pi,s}$ and defined as follows: 
    \begin{align*}
        \evaluate{x\star y}{\pi,s} &=
        \begin{cases}
        1 \quad \text{if } s(x)\star s(y),\\
        0 \quad \text{otherwise}, 
        \end{cases}
        \text{, where } \star\in\{=,\neq\}, \hspace{-8cm}\\
        \evaluate{R(\bar{x})}{\pi,s} &= \pi(R(s(\bar{x}))), &
        \evaluate{\neg R(\bar{x})}{\pi,s} &= \pi(\neg R(s(\bar{x}))),\\
        \evaluate{\phi\wedge\psi}{\pi,s} &= \evaluate{\phi}{\pi,s}\cdot\evaluate{\psi}{\pi,s}, &
        \evaluate{\phi\lor\psi}{\pi,s} &= \evaluate{\phi}{\pi,s}+\evaluate{\psi}{\pi,s}, \\
        \evaluate{\exists x\phi}{\pi,s} &= \sum_{a\in A}\evaluate{\phi}{\pi,s(a/x)}, & \evaluate{\forall x\phi}{\pi,s} &= \prod_{a\in A}\evaluate{\phi}{\pi,s(a/x)}, \\
        \evaluate{\phi\circ\psi}{\pi,s} &=
        \begin{cases}
        1 \quad \text{if } \evaluate{\phi}{\pi,s}\circ\evaluate{\psi}{\pi,s},\\
        0 \quad \text{otherwise}, 
        \end{cases}
        \text{, where } \circ\in\{=,\neq,\leq,\not\leq\}. \hspace{-4cm}
    \end{align*}
    If $\phi$ is a \emph{sentence}, we write $\evaluate{\phi}{\pi}$ for $\evaluate{\phi}{\pi,\emptyset}$.
\end{definition}
 We say that a $K$-interpretation $\pi$  is \textit{model-defining} \cite{gradel17semiringprovenancefirstordermodel}, if for all $R(\bar{{a}})$, we have $\pi(R(\bar{a}))=0$ iff $\pi(\neg R(\bar{a}))\neq 0$. \label{par:model_defining}
 
Sentences $\phi$ and $\psi$ of $\FO(O)$ are $K$-equivalent, written as $\phi\equiv_K \psi$, if $\evaluate{\phi}{\pi}=\evaluate{\psi}{\pi}$ for all 
$K$-interpretations. For logics $\mathcal{L}$ and $\mathcal{L}'$, we write $\mathcal{L}\leq_K\mathcal{L}'$ if for each sentence $\phi$ from $\mathcal{L}$ there is a sentence $\psi$ from  $\mathcal{L}'$ such that $\phi\equiv_K \psi$. If $\mathcal{L}\leq_K\mathcal{L}'$ and $\mathcal{L}'\leq_K\mathcal{L}$, we write $\mathcal{L}\equiv_K\mathcal{L}'$, and say that $\mathcal{L}$ and $\mathcal{L}'$ are \textit{equally expressive under} $K$. As usual, we write $\mathcal{L}<_K\mathcal{L}'$ if $\mathcal{L}\leq_K\mathcal{L}'$ and $\mathcal{L'}\not\leq_K\mathcal{L}$.

\begin{figure}
\centering
\begin{tikzpicture}[every node/.style={draw, rectangle, rounded corners, inner sep=1mm},x=2cm]
    \node (OSL) at (0, .5) {\footnotesize OSL};
    \node (GOT) at (1, 0) {\footnotesize GOT};
    \node (KI) at (-.2, -1.2) {\footnotesize KI};
    \node (FKH) at (-.8, -.5) {\footnotesize FKH};
    
    \draw (OSL) -- (KI) node[near end, left=0mm, draw=none, fill=none] {20};
    \draw (OSL) -- (FKH) node[midway, above=1mm, draw=none, fill=none] {10};
    \draw (GOT) -- (FKH) node[near start, above=0mm, draw=none, fill=none] {4};
    \draw (GOT) -- (KI) node[midway, below=1mm, draw=none, fill=none] {14};
    \draw (FKH) -- (KI) node[near start, below=1mm, draw=none, fill=none] {12};

\end{tikzpicture}
\caption{Ferry network as a $K$-interpretation. \label{fig:air}}
\end{figure}

\begin{example}\label{ex:ferry}
    Suppose $A$ consists of the city codes OSL (Oslo), GOT (Göteborg), KI (Kiel), FKH (Frederikshavn), and $\tau$ of one edge relation $E$. Consider the tropical semiring $\mathbb{T}=(\mathbb{R}\cup\{\infty\}, \min, +, \infty, 0)$.
    An example $\mathbb{T}$-interpretation $\pi\colon Lit_{A,\tau}\to \mathbb{T}$ is obtained by assigning $E(x,y)$, for each pair of cities $x$ and $y$, a number representing the duration of a direct ferry ride between them, as in \Cref{fig:air};
    if no direct ferry connection exists between two cities $x$ and $y$, $E(x,y)$ is assigned $\infty$. Furthermore each negated fact $\neg E(x,y)$ can be assigned a number so that $\pi$ is model-defining.
    Then, we can express in $\FO[<]$ that the duration of a direct ferry trip between any two cities $x$ and $y$ is shorter than the sum of durations for ferries from $x$ to  $z$ and $z$ to $y$, for any city $z$:
    \begin{equation}\label{eq:phi}
    \phi \coloneqq \forall xyz( E(x,y) <  (E(x,z) \land E(z,y))).
        \end{equation}

    Clearly, $\evaluate{\phi}{\pi}=0$, that is, $\phi$ evaluates to the identity element of multiplication under $\pi$.
\end{example}

We now consider some results that demonstrate the expressivity of the logics in the special case of the Boolean semiring $\mathbb{B}=(\mathbb{B},\lor,\wedge,0,1)$. Let $\mathcal{A}$ be a structure of vocabulary $\tau$.
The canonical truth interpretation $\pi_{\mathcal{A}}$, is the $\mathbb{B}$-interpretation $\pi\colon\lit_{A,\tau}\to \mathbb{B}$ defined for each $R(\bar{a})$ and $\neg R(\bar{a})$ as follows: $\pi(R(\bar{a}))=1$ if $\bar{a}\in R^{\mathcal{A}}$, $\pi(R(\bar{a}))=0$ if $\bar{a}\not\in R^{\mathcal{A}}$,
and $\pi(\neg R(\bar{a}))=1-\pi(R(\bar{a}))$.
The following proposition shows that for $\FO$-formulas, the canonical truth interpretation $\pi_{\mathcal{A}}$ under an assignment $s$ corresponds to the usual first-order formula evaluation under $s$ in the structure $\mathcal{A}$. 
It can be proven by induction on $\alpha$.
\begin{proposition}[\cite{gradel17semiringprovenancefirstordermodel}]
   Let $\alpha$ be an $\FO$-formula, and $\mathcal{A}$ a structure. Then  $\evaluate{\alpha}{\pi_{\mathcal{A}},s}=1$ if and only if $\mathcal{A}\models_s\alpha$.
\end{proposition}
Let $\xi_K\colon K\to\mathbb{B}$ be the characteristic mapping such that $\xi_K(a)=0$ if $a$ is the zero element of $K$ and $\xi_K(a)=1$ otherwise. Collapsing the interpretations of $\FO$-formulas to the Boolean semiring by taking the composition of the characteristic mapping $\xi_K$ and a $K$-interpretation $\pi$ preserves the 0-valued interpretations.
\begin{proposition}[\cite{BarlagHKPV23}]\label{equivprop}
    Let 
    be $\pi\colon \lit_{A,\tau}\to K$ an interpretation. Then for all $\alpha\in\FO$, $\evaluate{\alpha}{\xi_K\circ\pi}=1$ if and only if $\evaluate{\alpha}{\pi}\neq 0$ .
\end{proposition}


In case of the Boolean semiring $\mathbb{B}$ (with $0<1$), having access to formula (in)equality does not increase expressivity in the sense of the following proposition.
\begin{proposition}[\cite{BarlagHKPV23}]
$\FO\equiv_{\mathbb{B}}\FO(=,\neq,\leq,\not\leq)$
\end{proposition}
The above proposition does not hold for all $K$. In particular, if multiplication is \emph{non-idempotent} on some element $a\in K$ (i.e., $a\neq a \cdot a$), we can express properties in $\FO(=)$ or $\FO(\neq)$ that are not definable in $\FO$.
The same holds true for $\FO(<)$ or $\FO(\leq)$ if multiplication is \emph{expansive} on some $a\in K$ (defined as $a< a \cdot a$). 
\begin{proposition}
    Let $K$ be a semiring. Then
    \begin{itemize}
        \item $\FO<_{K}\FO(=)$ and $\FO<_{K}\FO(\neq)$ if multiplication is non-idempotent on some element of $K$, and
        \item $\FO<_{K}\FO(\leq)$ and $\FO<_{K}\FO(<)$ if multiplication is expansive on some element of $K$.
    \end{itemize}
\end{proposition}
\begin{proof}
    Consider the formula $\psi$ (\Cref{eq:phi}) from \Cref{ex:ferry}.
    We claim that $\psi$ cannot be translated to $\FO$. Suppose for a contradiction that we find $\psi'\in\FO$  such that $\psi\equiv_{K}\psi'$.
   Let $A=\{1,2,3\}$, and let $\tau=\{E\}$ for a binary relation symbol $E$. Consider two interpretations $\pi\colon\lit_{A,\tau}\to K$ and $\pi'\colon\lit_{A,\tau}\to K$ 
   such that for all $x,y\in A$, $\pi(E(x,y))=a$ and $\pi'(E(x,y))=1$, where $a$ is the non-idempotent element in $K$. Then $\evaluate{\psi}{\pi}=1$, because $\evaluate{E(x,y)}{\pi} = a < a \cdot a = \evaluate{E(x,z)\land E(z,y)}{\pi}$, for all $x,y,z\in A$. In contrast, using non-idempotency of $a$,  we obtain $\evaluate{\psi}{\pi'}=0$. 
   Consequently, our hypothesis entails $\evaluate{\psi'}{\pi}=1$ and $\evaluate{\psi'}{\pi'}=0$, so by Proposition~\ref{equivprop}, $\evaluate{\psi'}{\xi_{K}\circ\pi}= 1$ and $\evaluate{\psi'}{\xi_{K}\circ\pi'}= 0$. But since $\xi_{K}\circ\pi=\xi_{K}\circ\pi'$, this is impossible. Thus the claim follows.
   It is straightforward to verify that substituting any of $\neq,=,\geq$ for $<$ in $\psi$ results in the same outcome (specifically $\leq$ needs to swapped). 
   \qed
\end{proof}
As a consequence, we note that the property considered in \Cref{ex:ferry} is not expressible in $\FO$ over $\mathbb{T}$-interpretations,
as the multiplication of the tropical semiring $\mathbb{T}$ is expansive on some element.

\subsection{Encodings of $K$-interpretations}

In order to compare $\FO(O)$ to the machine models we introduced, we need to identify it with a fitting set of functions.
For any $O \subseteq \{=, \neq, \leq, \not\leq\}$, any $\FO(O)$ sentence can essentially be seen as a function from the set of $K$-interpretations to $K$.
To make this fit in with our machine models, we define an encoding for $K$-interpretations, so that the function defined by an $\FO(O)$ sentence can be seen as function from $K^*$ to $K$.

\begin{definition}
    Let $A$ be a set with a strict total order, let $\tau$ be a relational signature, let $\lit_{A, \tau} = \{\ell_1, \dots, \ell_n\}$, let $K$ be a  
    semiring and let $\pi \colon \lit_{A, \tau} \to K$ be a $K$-interpretation.
    Then we define $\enc(\pi)$ to be the concatenation of the values assigned to each literal by $\pi$, i.e., 
    \[
        \enc(\pi) \coloneqq (\pi(\ell_1), \dots, \pi(\ell_n)) \in K^n.
    \]
    For technical reasons, we encode literals of relation symbols $R$ of arity $0$ as if they had arity $1$, i.e, as $\lvert A \rvert$ copies of $\pi(R())$.\label{def:enc}
\end{definition}

Of particular interest is that we can determine $\lvert A \rvert$ from $\lvert \enc(\pi) \rvert$.
With this minor technical change at hand, we can now compute $\lvert A \rvert$ from $\lvert \enc(\pi) \rvert$ and $\tau$.

\begin{lemma}\label{lem:decode_A_size}
    Let $A$ be a strictly ordered set, let $\tau$ be a relational vocabulary, let $K$ be a semiring and $\pi \colon \lit_{A, \tau} \to K$ be a $K$-interpretation.
    We can compute $\lvert A \rvert$ when given $\lvert \enc(\pi) \rvert$ in logarithmic time on a $\BSSK$ machine.
\end{lemma}
\begin{proof}
We can use, e.g., binary search to find the solution for $\lvert A \rvert$ in
\[
    \lvert \enc(\pi) \rvert = \sum_{R \in \tau} \lvert A \rvert^{\max(\ar(R), 1)} \cdot 2
\]
in logarithmic time.
\qed
\end{proof}

\subsection{$\FO$ with Built-in Relations}

To characterize circuit classes logically later on, we need to extend this logic by additional ``built-in'' $K$-relations that are not part of the $K$-interpretation. 
To that end we will slightly extend the syntax of $\FO(O)$ to $\FO(O, F)$ for particular function families $F$.
We essentially want to allow additional $K$-relations that may depend on the size of $A$, but not on $A$ itself. 
We therefore treat ordered sets of the same cardinality as isomorphic to the first $\lvert A \rvert$ natural numbers and thus define the aforementioned function families accordingly.
The set $\arb$ is the set of all function families of the aforementioned kind (cf., e.g.,  \cite{HeribertBuch} for the  analogous definition of built-in relations in the standard Tarski semantics). 

\begin{definition}
Let $K$ be a semiring. 
Then 
\[
    \arb \coloneqq \bigcup_{k \in \N} \{(f_n)_{n \in \N} \mid \text{$f_n \colon \{1, \dots, n\}^k \to K$ for all $n \in \N$}\}.
\]
\end{definition}
We also need to extend the notion of a signature to allow to differentiate between built-in $K$-relations and those given by the input $K$-interpretation. 
To that end, we adopt an auxiliary vocabulary $\sigma_F$ consisting of a relation symbol $P_f$ for each $f=(f_n)_{n \in \N}$ in $F$.

\begin{definition}
    Let $O \subseteq \{=, \neq, \leq, \not\leq\}$, let $\tau$ be a relational vocabulary and let $K$ be a  semiring.
    Then for any $F \subseteq \arbF$ we define the syntax of $\FO(O, F)$ $\tau$-formulae  by extending the Backus--Naur form for $\FO(O)$ formulae by the rule
    $
        \alpha \Coloneqq P_f(\ol{x}) 
    $,
    where $P_f \in \sigma_F$ and $\ol{x}$ is a tuple of variables such that $\lvert \ol{x} \rvert = \ar(P)$.
\end{definition}

The semantics of $\FO(O, F)$ are now defined in the obvious way for ranked (ordered) structures.
\begin{definition}


    Let $O \subseteq \{=, \neq, \leq, \not\leq\}$, let $K$ be a semiring, let $F \subseteq \arbF$, and let $\tau$ be a finite vocabulary. 
      Furthermore let $A$ be a strictly ordered set, $\pi \colon \lit_{A, \tau} \to K$  a $K$-interpretation, and $s$  an assignment.
    Then the semantics of $\FO(O, F)$ extend the semantics for $\FO(O)$ by 
  $$\evaluate{P_f(\ol{x})}{\pi,s} = f_{|A|}(r(s(\ol{x}))),$$
    where $P_f \in \sigma_F$, $\ol{x}$ is a tuple of variables such that $\lvert \ol{x} \rvert = \ar(P)$ and $r \colon A \to \{1, \dots, \lvert A \rvert\}$ is the \emph{ranking function} on $A$, which maps each element of $A$ to its position in the ordering on $A$.

\end{definition}

With this definition at hand, we can finally define the set of functions definable by $\FO(O, F)$ sentences.
Note that we omit $O$ (resp. $F$), if it is empty.

In the following, for $\pi$, we will use $\enc(\pi)$ specified in Definition~\ref{def:enc} whenever we consider problems that use $\pi$ as part of the input.

\funcproblemdef{$\FOK(O,F)\EVAL_\varphi$, $O \subseteq \{=, \neq, \leq, \not\leq\}$,  $F \subseteq \arbF$, semiring $K$}{$K$-interpretation $\pi$}{$\evaluate{\varphi}{\pi}$}

To denote the set of all these function problems, we use the following notation.

\begin{definition}
Let $O \subseteq \{=, \neq, \leq, \not\leq\}$, $K$ be a semiring, and $F \subseteq \arbF$. 
Then
\[
\begin{array}{r@{\,}l}
        \FOK(O, F) &\coloneqq \{\,\FOK(O, F)\EVAL_\varphi\mid \varphi \in \FO(O, F)\,\}.
\end{array}
\]
Again, we omit $O$ (resp. $F$) from the notation if it is empty.
\end{definition}


In the upcoming sections, we establish several connections between the previously introduced models of computation and logic.
It is noteworthy that our results generalize beyond model-defining $K$-interpretations, as defined on page~\pageref{par:model_defining}.

\section{The Complexity of Model Checking for \texorpdfstring{$\FOK$}{FOK}}\label{results}

In the following, we define the model checking problem. 

\decproblemdef{
    $\FOKMC{}$, $O\subseteq\{=,\neq,\leq,\not\leq\}$, semiring $K$
}{
    $\FOK(O)$ formula $\varphi$, $K$-interpretation $\pi$, assignment $s$
}{
    $\evaluate{\varphi}{\pi,s}\neq0$
}
Given a formula $\varphi$,
we write $\FOKMC{\varphi}$ to denote the fragment of 
$\FOKMC{}$ restricted to the fixed formula $\varphi$.

Regarding input values, we assume for $\varphi$, $s$, and $A$ standard polynomial-time computable encodings, e.g., binary encoding. 

The proof of the next result can be found in Appendix~\ref{prf:eval-pspace-algo}.
\begin{theorem}
    Fix a semiring $K$. 
    Given an $\FOK$ formula $\varphi$, a $K$ interpretation $\pi$, and an assignment $s$. 
    The value $\evaluate{\varphi}{\pi,s}$ can be computed in time $\bO(n^2 \cdot |\varphi| \cdot |\pi|^{|\varphi|})$ and space in $\bO(\mathrm{poly}(n))$, with $n = |\varphi|+|\pi|+|s|$ the sum of the encoding lengths.\label{thm:eval-pspace-algo}
\end{theorem}

\begin{corollary}
Let $O \subseteq \{=, \neq, \leq, \not\leq\}$. 
    Every $f\in\FOK(O)$ can be computed in polynomial space.
\end{corollary}
It is easy to check that the above corollary can be extended to functions $f\in\FOK(O,F)$, where  $F \subseteq \arbF$, assuming each  $(f_n)_{n \in \N}\in F$ is computable in polynomial space.

The following two corollaries are obtained via utilisation of Algorithm~\ref{alg:evaluation of varphi} from Appendix~\ref{prf:eval-pspace-algo} and merely checking whether the computed value of $\evaluate{\varphi}{\pi,s}$ is not $0$.
\begin{corollary}
     Let $O\subseteq\{=,\neq,\leq,\not\leq\}$. For any $\phi\in\FO(O)$, $\FOKMC{\varphi} \in \PK$.
\end{corollary}

\begin{corollary}
    Let $O\subseteq\{=,\neq,\leq,\not\leq\}$.
    $\FOKMC{} \in \PSPACEK$.
\end{corollary}

\section{A Circuit Characterisation of \texorpdfstring{$\FOK$}{FOK}}\label{circ}

The following adapts a result from the Boolean setting, originally due to Immerman~\cite{DBLP:journals/siamcomp/Immerman87}, which formalises the intuition that first-order logic and constant-depth circuits correspond to one another.
More recently, this result has been generalised to metafinite logics over the reals~\cite{10.1093/logcom/exae051} and more general integral domains~\cite{BCG24}.
We establish a similar result, moving from integral domains to semirings and replacing logics over metafinite structures with a logic that is evaluated directly in the semiring.

\begin{theorem}\label{thm:fo_aco}
Let $O \subseteq \{=, \neq, \leq, \not\leq\}$ and let $K$ be a semiring. 
Then for $K$-interpretations $\pi \colon \lit_{\tau, A} \to K$, where $A$ is strictly ordered: $\FOK(O, \arb) = \FACO[O]$.
\end{theorem}

\begin{proof}
    This proof follows the same general pattern as a similar result about circuits over the reals and first-order logic over metafinite $\R$-structures \cite[Theorem~30]{10.1093/logcom/exae051}.
    
    The basic idea is for the direction $\FOK(O, \arb) \subseteq \FACO[O]$ to define a circuit family that uses arithmetic to mimic the semantics of a first-order formula $\varphi$, where the input to the circuit family is the encoding of a $K$-interpretation of $\varphi$.
    Existential (resp. universal) quantification can be represented by unbounded addition (resp. multiplication) gates, logical connectives can be simulated by the respective bounded gate types, $K$-relations represent elements of the input, elements of $O$ are given by the extension and (in)equalities on the level of variables can be hardcoded, since we are considering sentences.
    
    For the converse direction, we define a sentence that essentially describes the way a circuit is evaluated, using the additional built-in relations to describe the structure of the circuit.
    We basically define a recursive formula that yields the value of gates at a particular depth in the circuit.
    Since the depth of the circuit family is constant, unrolling the recursion yields a formula that is independent of the particular circuit and only depends on the (constant) depth of its family.

    The full proof can be found in Appendix~\ref{proof:thm_fo_aco}.
    \qed
\end{proof}



\section{Conclusion}


In this paper, we introduced several models of computation to analyze the complexity of problems with respect to semirings.
In particular, we adapted BSS machines and arithmetic circuits for semirings to generalize previously established models for computation with fields or rings.
We then characterized the complexity of the model checking and evaluation problems of first-order logic with semiring semantics using these models.
The work in establishing a complexity theory started here gives rise to an abundance of further research directions.

Continuing from the model checking question, other possible connections between semiring logics and sequential computation merit investigation. 
In particular, the well-known theorem by Fagin, establishing a connection between second-order logic and NP~\cite{fagin1974generalized}, which has been adapted to BSS machines and logics over the real numbers by Grädel and Meer~\cite{DBLP:conf/stoc/GradelM95}, warrants analysis with respect to semirings.

Furthermore, there is much work to be done with regard to arithmetic circuits over semirings. 
The result shown in this paper only pertains to so-called \emph{non-uniform} circuit families, meaning circuit families, where there is no restriction on how computationally difficult it is to obtain any individual circuit.
In general, this can lead to problems solvable by such circuit families, that are not computable with regard to BSS machines.
In order to view a circuit family as an algorithm, a restriction on how hard it is to obtain any given circuit is required.
Given the constructive nature of our proof, there is no doubt that it can be made uniform.
The exact nature of that uniformity still needs to be examined, however.
Additionally, larger circuit classes than $\FACO$ could be characterized, following, e.g., the characterization to the entire $\mathrm{AC}$ and $\mathrm{NC}$ hierarchies over the reals~\cite{BCG24}.
It also remains an open question how the extensions of 
$\FO$ considered in this work relate to logics over metafinite structures, and what their finite-model-theoretic properties are.

\bibliography{ref_url,references}

\newpage
\appendix
\section{BSS$_K$ Machines} \label{appendix}
For the convenience of the reader, we repeat the definition of a $\BSSK$-machine in a more detailed fashion.
\begin{definition}[$\BSSK$ machines] \label{def:BSSK} 
    Let $K$ be a
    semiring. 
    A $\BSSK$ machine consists of an input space $\mathcal{I}=K^*$, a state space $\mathcal{S}=K_*$ and an output space $\mathcal{O}=K^*$, together with a
    directed graph whose nodes are labelled by $1, \ldots, N$. 
    The nodes are of five different types.
    \begin{itemize}
        \item \emph{Input node}. The node labelled by $1$ is the only input node. The node is associated with a next node $\beta(1)$ and the input mapping $g_I\colon \mathcal{I} \to \mathcal{S}$.
        \item \emph{Output node}. The node labelled by $N$ is the only output node. This node is not associated with any next node. 
        Once this node is reached, the computation halts, and the result of the computation is placed on the output space by the output mapping $g_O\colon \mathcal{S}\to \mathcal{O}$.
        \item \emph{Computation nodes.} A computation node $m$ is associated with a next node $\beta(m)$ and a mapping $g_m\colon \mathcal{S}\to \mathcal{S}$ such that for some $c \in K$ and $i,j,k\in \mathbb{Z}$ the mapping $g_m$ is identity on coordinates $l \neq i$ and on coordinate $i$ one of the following holds:
        \begin{itemize}
            \item $g_m(x)_i =x_j+x_k$ (addition),
            \item $g_m(x)_i = x_j\cdot x_k$ (multiplication),
            \item $g_m(x)_i = c$ (constant assignment).
        \end{itemize} 
        \item \emph{Branch nodes.} A branch node $m$ is associated with nodes $\beta^-(m)$ and $\beta^+(m)$. 
        Given $x\in \mathcal{S}$ the next node is $\beta^-(m)$ if $x_1 = x_2$ and $\beta^+(m)$, otherwise. 
        If $K$ is ordered, instead the next node is $\beta^-(m)$ if $x_1 \leq x_2$ and $\beta^+(m)$, otherwise.
        \item \emph{Shift nodes.} A shift node $m$ is associated either with shift left $\sigma_l$ or shift right $\sigma_r$, and a next node $\beta(m)$.
    \end{itemize}
    The input mapping $g_I\colon \mathcal{I} \to \mathcal{S}$ places an
    input $(x_1, \ldots ,x_n)$ in the state
    \[
        (\ldots , 0, \underbrace{1, \dots, 1}_{n}, 0\textbf{.} x_1, \ldots ,x_n, 0,\ldots )\in \mathcal{S},
    \]
    where the size of the input $n$ is encoded in unary in the $n$ first negative coordinates.
    The output mapping $g_O\colon \mathcal{S}\to \mathcal{O}$ maps a state to the string consisting of its first $l$ positive coordinates, where $l$ is the number of consecutive ones stored in the negative coordinates starting from the first negative coordinate.
    For instance,
    \[
        g_O((\ldots ,2,1,1,1,0\textbf{.}x_1, x_2,x_3,x_4,\ldots )) = (x_1, x_2,x_3).
    \]
    A configuration at any moment of computation consists of a node $m\in \{1, \ldots ,N\}$ and a current state $x\in\mathcal{S}$.
    The (sometimes partial) \emph{input-output} function $f_M \colon K^*\to K^*$ of a machine $M$ is now defined in the obvious manner.
    A function $f \colon K^*\to K^*$ is \emph{computable} if $f=f_M$ for some machine $M$. A language $L \subseteq K^*$ is \emph{decided} by a $\BSSK$ machine $M$ if its characteristic function $\chi_L \colon K^* \to K^*$ is $f_M$. 
\end{definition}




Arbitrary computations and comparisons need a way to remember intermediate values.
To this end, we use $(\dots, 0\textbf{.}\, x_1, 1, x_2, 1 \dots, x_n, 1, 0 \dots)$ as a more useful starting state of a $\BSSK$ machine.
The gap between two input values allows the machine to temporarily store values or move a value from one end of the state space to the other.
The following result shows, that obtaining this \emph{gap normal form} can be achieved with polynomial
overhead as an initialization step of a $\BSSK$ machine.


\begin{proposition}\label{prop:gap}
    The initial state space $(\dots, 0, 1, \ldots, 1, 0\textbf{.} x_1, \ldots ,x_n, 0,\ldots )$ can be converted into $(\dots, 0\textbf{.}\, x_1, 1, x_2, 1 \dots, x_n, 1, 0 \dots)$ with quadratic overhead.
    Similarly, the reverse direction can be achieved with quadratic overhead.
\end{proposition}
\begin{proof}
Figure~\ref{fig:bss normal form machine} implements the subroutine \texttt{Init}, converting the initial configuration into $(\dots, 0\textbf{.}\, x_1, 1, x_2, 1 \dots, x_n, 1, 0 \dots)$ in $\bO(n^2)$ steps.
Analogously, Figure~\ref{fig:bss normal form machine reverse} implements the reverse.
\qed
\end{proof}

\begin{figure}[!h] 
    \centering
    \begin{tikzpicture}[
        bssnode/.style = {draw, rounded corners},
        node distance=.66
        ]
        \node (0) {$(\dots, 0, 1, \dots, 1, 0\textbf{.}\, x_1, \dots, x_n, 0, \dots)$};
        \node[bssnode, below = of 0] (1) {$\sigma_r$};
        \node[bssnode, below = of 1] (2) {$\sigma_r$};
        \node[bssnode, below = of 2] (3) {$\sigma_r$};
        \node[bssnode, below = of 3] (4) {$x_1 = x_2$?}; 
        \node[bssnode, below = of 4, align=left] (5) {$x_2 \gets x_4$\\$x_3 \gets 1$\\$x_4 \gets 0$};
        \node[bssnode, below = of 5] (6) {$\sigma_r$};
        \node[bssnode, below = of 6] (7) {$x_1 = x_2$?}; 
        \node[bssnode, below = of 7, align=left] (8) {$x_2 \gets x_3$\\$x_3 \gets 1$};
        \node[bssnode, below = of 8] (9) {$\sigma_r$};
        \node[bssnode, right = of 8, align=left] (10) {$x_1 \gets x_3$\\$x_3 \gets 0$};
        \node[bssnode, below = of 10] (12) {$\sigma_l$};
        \node[bssnode, below = of 12] (12 1) {$\sigma_l$};
        \node[bssnode, below = of 12 1] (12 2) {$\sigma_l$};
        \node[bssnode, below = of 12 2] (13) {$x_1 = x_2$?}; 
        \node[bssnode, below = of 13] (13 1) {$\sigma_r$};
        
        \node[bssnode, right = 4cm of 5, align=left] (14) {$x_1 \gets x_4$\\$x_3 \gets 1$\\$x_4 \gets 0$};
        
        \node[bssnode, below = of 14] (15) {$\sigma_l$};
        \node[bssnode, below = of 15] (17) {$x_1 = x_2$?}; 
        \node[bssnode, below = of 17] (18) {$\sigma_r$};
        \node[bssnode, below = of 18, align=left] (19) {$x_2 \gets x_1$\\$x_1 \gets 1$};
        \node[bssnode, below = of 19] (20) {$\sigma_r$};
        
        \node[bssnode, right = of 18] (21) {$x_2 \gets 0$};
        \node[bssnode, below = of 21] (22) {$\sigma_l$};
        \node[bssnode, below = of 22] (23) {$\sigma_l$};
        \node[below = of 23] (24) {$(\dots, 0\textbf{.}\, x_1, 1, x_2, 1 \dots, x_n, 1, 0 \dots)$};
        
        \draw[-stealth'] (0) -- (1);
        \draw[-stealth'] (1) -- (2);
        \draw[-stealth'] (2) -- (3);
        \draw[-stealth'] (3) -- (4);
        \draw[-stealth'] (4) -- node[right]{yes} (5);
        \draw[-stealth'] (5) -- (6);
        \draw[-stealth'] (6) -- (7);
        \draw[-stealth'] (7) -- node[right]{yes} (8);
        \draw[-stealth'] (8) -- (9);
        \draw[-stealth'] (9) -- ++(-1.25,0) |- (7);
        \draw[-stealth'] (7) -- node[above]{no} ++(1.5,0) -| (10);
        \draw[-stealth'] (10) -- (12);
        \draw[-stealth'] (12) -- (12 1);
        \draw[-stealth'] (12 1) -- (12 2);
        \draw[-stealth'] (12 2) -- (13);
        \draw[-stealth'] (13) -- node[above]{yes} ++(-1.5,0) |- (12 2);
        \draw[-stealth'] (13) -- node[right]{no} (13 1);
        \draw[-stealth'] (13 1) -- ++ (-4,0) |- (4);
        \draw[-stealth'] (4) -- node[above]{no} ++(1.5,0) -| (14);
        \draw[-stealth'] (14) -- (15);
        \draw[-stealth'] (15) -- (17);
        \draw[-stealth'] (17) -- node[right]{yes} (18);
        \draw[-stealth'] (17) -- node[above]{no} ++(1.5,0) -| (21);
        \draw[-stealth'] (18) -- (19);
        \draw[-stealth'] (19) -- (20);
        \draw[-stealth'] (20) -- ++(-1.5,0) |- (17);
        \draw[-stealth'] (21) -- (22);
        \draw[-stealth'] (22) -- (23);
        \draw[-stealth'] (23) -- (24);
    \end{tikzpicture}
    \caption{\texttt{Init} subroutine. 
    Converts an input of a $\BSSK$ machine into gap normal form. 
    The elements on the right of the dot are always $x_1, x_2, \dots$ during the computation. 
    And $\sigma_l$ (resp. $\sigma_r$) shift the \textbf{state space} to the left (right) with respect to the dot.    
    The intuition of the algorithm is to iteratively pair each $x$ value with a 1.
    Also the shift and compare nodes are placed in such a way that no comparison with the input is made.
    This avoids problems when the input has 0 or 1 values.
    }
    \label{fig:bss normal form machine}
\end{figure}
\begin{figure}
    \centering
    \begin{tikzpicture}[
        bssnode/.style = {draw, rounded corners},
        node distance=.66
        ]
        \node (0) {$(\dots, 0\textbf{.}\, x_1, 1, x_2, 1 \dots, x_n, 1, 0 \dots)$};
        \node[bssnode, below = of 0, align=left] (1) {$x_0 \gets x_1$\\$x_1 \gets 1$};
        \node[bssnode, below = of 1] (2) {$x_1 = x_2$?};
        \node[bssnode, below = of 2] (3) {$\sigma_l$};
        \node[bssnode, below = of 3] (4) {$\sigma_l$};
        \node[bssnode, left = 1cm of 2] (5) {$\sigma_r$};
        \node[bssnode, below = of 5] (6) {$\sigma_r$};
        \node[bssnode, below = of 6] (7) {$\sigma_r$};
        \node[bssnode, below = of 7, align=left] (8) {$x_4 \gets x_1$\\$x_1 \gets 0$};
        \node[bssnode, below = of 8] (9) {$x_1 = x_2$?};
        \node[bssnode, below = of 9] (10) {$\sigma_r$};
        \node[bssnode, right = 1cm of 9] (11) {$\sigma_r$};
        \node[bssnode, below = of 11] (12) {$x_1 = x_2$?};
        \node[bssnode, below = of 12, align=left] (13) {$x_2 \gets x_0$\\$x_0 \gets 0$};
        \node[bssnode, below = of 13] (14) {$\sigma_l$};        
        \node[bssnode, below = of 14] (15) {$\sigma_l$};
        \node[bssnode, below = of 15] (16) {$x_1 = x_2$?};
        \node[bssnode, below = of 16, align=left] (17) {$x_1 \gets x_0$\\$x_0 \gets 1$};
        \node[bssnode, left = 1cm of 16, align=left] (18) {$x_2 \gets x_0$\\$x_1 \gets 0$\\$x_0 \gets 1$};
        \node[bssnode, below = of 18] (19) {$\sigma_r$};
        \node[bssnode, below = of 19] (20) {$\sigma_r$};
        
        \node[bssnode, right = 1.75cm of 12] (21) {$\sigma_l$};
        \node[bssnode, below = of 21] (22) {$\sigma_l$};
        \node[bssnode, below = of 22] (23) {$x_1 = x_2$?};
        \node[bssnode, below = of 23] (24) {$\sigma_l$};
        \node[bssnode, right = 0.75cm of 23] (25) {$\sigma_l$};
        \node[bssnode, below = of 25] (26) {$\sigma_l$};        
        \node[below = of 26] (27) {$(\dots, 0, 1, \dots, 1, 0\textbf{.}\, x_1, \dots, x_n, 0, \dots)$};
        
        \draw[-stealth'] (0) -- (1);
        \draw[-stealth'] (1) -- (2);
        \draw[-stealth'] (2) -- node[right]{yes} (3);
        \draw[-stealth'] (3) -- (4);
        \draw[-stealth'] (4) -- ++(1.25,0) |- (1);
        \draw[-stealth'] (2) -- node[above]{no} (5);
        \draw[-stealth'] (5) -- (6);
        \draw[-stealth'] (6) -- (7);
        \draw[-stealth'] (7) -- (8);
        \draw[-stealth'] (8) -- (9);
        \draw[-stealth'] (9) -- node[above]{no} (11);
        \draw[-stealth'] (11) -- (12);
        \draw[-stealth'] (12) -- node[right]{no} (13);
        \draw[-stealth'] (9) -- node[right]{yes} (10);
        \draw[-stealth'] (10) -- ++(-1.25,0) |- (9);
        \draw[-stealth'] (13) -- (14);
        \draw[-stealth'] (14) -- (15);
        \draw[-stealth'] (15) -- (16);
        \draw[-stealth'] (16) -- node[right]{yes} (17);
        \draw[-stealth'] (16) -- node[above]{no} (18);
        \draw[-stealth'] (17) -- ++(1.25,0) |- (15);
        \draw[-stealth'] (18) -- (19);
        \draw[-stealth'] (19) -- (20);
        \draw[-stealth'] (20) -- ++(-1.25,0) |- (9);
        \draw[-stealth'] (12) -- node[above]{yes} (21);
        \draw[-stealth'] (21) -- (22);
        \draw[-stealth'] (22) -- (23);
        \draw[-stealth'] (23) -- node[right]{yes} (24);
        \draw[-stealth'] (24) -- ++(-1.25,0) |- (23);
        \draw[-stealth'] (23) -- node[above]{no} (25);
        \draw[-stealth'] (25) -- (26);
        \draw[-stealth'] (26) -- (27);
    \end{tikzpicture}
    \caption{Reverse \texttt{Init} subroutine, i.e., \texttt{Init$^{-1}$}.
    Converts the gap normal form into the input/output form of a BBS$_K$ machine. 
    The elements on the right of the dot are always $x_1, x_2, \dots$ during the computation. 
    And $\sigma_l$ (resp. $\sigma_r$) shift the \textbf{state space} to the left (right) with respect to the dot.    
    The intuition of the algorithm is to iteratively decouple $(x_i, 1)$ pairs.
    Also the shift and compare nodes are placed in such a way that no comparison with the $x_i$ values is made.
    This avoids problems when they contain 0 or 1 values.
    }
    \label{fig:bss normal form machine reverse}
\end{figure}

\begin{remark}\label{remark:comparisons}
    Using the gap normal form allows easy simulation of branching nodes using $x_i = x_j$ or $x_i = c$ as branching condition.
\end{remark}

\section{Proof of Theorem~\ref{lemma:turing}}\label{proof:turing}

\newcommand{\hatzero}{\hat{0}}
\newcommand{\hatone}{\hat{1}}
\newcommand{\hata}{\hat{a}}
\newcommand{\hatblank}{\varepsilon}

\begin{restate}[\ref{lemma:turing}]
\begin{theorem}
    Let $K$ be a semiring and
    let $f \colon K^* \to K^*$ be a function
    computed by a $K$-Turing machine $M$
    that runs in time $t$ and space $s$.
    Then there is a $\BSSK$ machine $M'$ computing $f$ and
    $c \in \N$ such that
    for each $x \in K^*$, the machine $M'$ runs on input $x$ in time $c(t(|x|) + |x|^2 + |f(x)|^2 + 1)$
    and in space $c(s(|x|) + 1)$.
\end{theorem}
\end{restate}

\begin{proof}
Let $M = (Q, q_0, R, P, \Gamma, b, K, \delta)$.
We give an informal description of the
simulation of the $K$-Turing machine $M$ using a $\BSSK$ machine.
The simulation proceeds through three consecutive phases.
\begin{enumerate}
\item The input $x \in K^*$ of $M$ is first converted to the gap normal form using
Proposition~\ref{prop:gap}, thus allowing Remark~\ref{remark:comparisons}
to be applied. This intermediate string is then converted to another normal form
that allows the machine $M'$ to simulate the use of the tape alphabet $\Gamma$ along with
the $K$-valued registers of $M$. To this end, the tape of the machine is conceptually divided into blocks
of $2k$ consecutive cells, in effect, allowing the tape
to be used in the form $(K^{2k})_*$, corresponding to shift operations of length $2k$ for
the underlying $\BSSK$ machine.
In total, the required conversions incur a quadratic overhead in
time and a linear overhead in space, based on the length of the input $x$.
\item The computation of the machine $M$ is simulated step by step according to the
transition function $\delta$, using the simulated state space $(K^{2k})_*$. Based on the fixed length
of $2k$ cells of the extended blocks, each simulated computation step of the $K$-TM $M$
incurs a constant overhead in time for the underlying $\BSSK$ machine,
yielding a total running time linear in $t(|x|)$ and space usage linear in $s(|x|)$
during the simulation phase.
\item Once the simulation is completed, the encoded output string corresponding to $f(x)$
in the state space $(K^{2k})_*$ is
first converted to the gap normal form,
and then to the output format of a $\BSSK$ machine by evoking Proposition~\ref{prop:gap}. These two
conversions can be achieved in quadratic time and linear space based on the length of
the output $f(x)$.
\end{enumerate}
All in all, the use of time and the required amount of space satisfy the conditions stated
in the claim.
In the following, we consider some of the details.

Once the input is converted to the gap normal form in Phase 1,
the machine can be thought of as if it was
using an alphabet of the form $\{\hata \mid a \in K\} \cup \{\hatblank\}$, where
for each element $a \in K$, the symbol $\hata$ encodes the semiring value $a$ in a block of
two consecutive cells of the tape in the form $a1$, and furthermore, where
$\hatblank$ carries the meaning of the blank symbol of the alphabet and is
encoded with two consecutive $0$ elements, i.e., the string $00$.
As in Remark~\ref{remark:comparisons}, the underlying $\BSSK$ machine
can simulate an extended $\BSSK$ machine using the alphabet
$\{\hata \mid a \in K\} \cup \{\hatblank\}$ and a conceptual state space $(K^2)_*$
in such a manner that the simulated machine is capable of computing the arithmetic operations of the semiring,
comparing elements on the tape with constant values of the semiring,
as well as comparing non-consecutive elements at fixed coordinates of the tape with each other.
Each of these operations can be implemented using a constant number of computation steps
on the underlying $\BSSK$ machine.

In particular, the gap normal form permits the use of symbols $\hatzero$ and $\hatone$
for binary encoding.
We exploit this fact in order to convert the set $\Gamma$ of tape alphabet symbols
into strings of a fixed length $l$ using the alphabet
$\{\hata \mid a \in K\} \cup \{\hatblank\}$.
These strings in turn correspond to strings of length $2l$
in the state space of the underlying $\BSSK$ machine.
In addition to the symbols of the set $\Gamma$, the encoding that is used in Phase 2
allows the elements of $K$ to be used as tape symbols and
simulates the use of registers corresponding to the set $R$.
Overall, this results in an encoding using strings of elements of $K$ with a fixed
length $2k$.

Next, we give an exact definition for this encoding.
Let, e.g., $l = \lceil\log_2(|\Gamma|)\rceil$, $k = l + |R| + 1$, and let
$e \colon \Gamma \setminus \{b\} \to \{\hatzero, \hatone\}^l$ be an injection.
Each element of the set $\Gamma \cup K$ is matched with a (possibly infinite) set
of length $k$ strings of the alphabet
$A \coloneqq \{\hata \mid a \in K\} \cup \{\hatblank\}$ as follows:
\begin{itemize}
\item The blank symbol $b \in \Gamma$ corresponds to every string $\hatblank s$ with $s \in A^{k - 1}$.
\item Each non-blank tape symbol $a \in \Gamma \setminus \{b\}$ corresponds to the concatenation $\hatzero e(a) s$
for every $s \in A^{|R|}$.
\item Each semiring element $a \in K$ corresponds to every $\hatone \hata s$, where $s \in A ^ {k - 2}$.
\end{itemize}
In particular, the previous three types of encoded symbols can be distinguished by the first $A$-element of each sequence.

The last $|R|$ elements of these strings are used to carry the $K$-values of
the registers in such a manner that whenever the extended head of the machine
in the simulated state space $(A^k)_*$ (or, more generally, $((K^2)^k)_*$) is shifted left or right, corresponding to a
shift of length $2k$ for the underlying $\BSSK$ machine, the elements stored in the
simulated registers are copied to their respective places in the new position of the simulated head.
In essence, the values in the registers are carried along as the head of the machine is moved.

Since the length of the extended cells of the state space $(A^k)_*$ is fixed,
each operation type of the transition function $\delta$ can be simulated using a fixed
number of computation steps of the underlying $\BSSK$ machine.
The states in the finite set $Q$ of the $K$-Turing machine $M$ are kept track of using
the nodes that are labelled by $1, \dots, N$ in Definition \ref{def:BSSK}.

As the last part of the proof, we explain in short how to implement the remaining
conversions in Phases 1 and 3.
In order to convert a string that is in the gap normal form to the encoding used in the simulation,
we repeatedly apply the procedure of replacing the rightmost $\hatzero$ or $\hatone$
by $\hatblank$ and moving the replaced symbol step by step to be the leftmost element corresponding to
the string of the simulation alphabet. During the process,
the remaining part of the string of symbols in the alphabet $\{\hata \mid a \in K\}$
and the converted part are separated using
the string $\hatblank^k$.
During the conversion, the extended blank symbols in the string $\hatblank^k$ can be recognized by
both sides of the simulation, thus allowing
this string to be used as a separator.
The conversion back to the gap normal form in Phase 3 can be implemented in a similar manner.
Both of these conversions can be accomplished in quadratic time and linear space
based on the length of the string that is to be converted.
\qed
\end{proof}

\section{Proof of Theorem~\ref{thm:fo_aco}}

In the sequel, when characterizing circuit classes logically, it will be necessary, to be able to refer to properties of individual gates.
For this reason, we follow the type associations defined in Tab.~\ref{table:gate_types}.

\begin{table}
\centering
    \begin{tabular}{cccccccccc}
        \toprule
        $g$ & input & constant & $+$ & $\cdot$ & output & $=$ & $\neq$ & $\leq$ & $\not\leq$\\
        \midrule
        type & 1 & 2 & 3 & 4 & 5 & 6 & 7 & 8 & 9\\
        \bottomrule
    \end{tabular}
    \caption{Type associations.}\label{table:gate_types}
\end{table}

\begin{restate}[\ref{thm:fo_aco}]
\begin{theorem}
    Let $O \subseteq \{=, \neq, \leq, \not\leq\}$ and let $K$ be a semiring. 
    Then for $K$-interpretations $\pi \colon \lit_{\tau, A} \to K$, where $A$ is strictly ordered: $\FOK(O, \arb) = \FACO[O]$.
\end{theorem}
\end{restate}
\begin{proof}\label{proof:thm_fo_aco}
$\FOK(O, \arb) \subseteq \FACO[O]$:

Given a $\FO(O, \arb)$-sentence $\varphi$, the idea is to construct a circuit family that computes the function $\FOK(O, \arb)\EVAL_\varphi$.
This is achieved by structural induction on $\varphi$. 
For any $n \in \N$, we essentially built the $n$th circuit ``top-down'', i.e., starting at the output gate and continuing towards the input gates.
While doing so, we maintain the invariant at each step in the induction, that if the predecessor gates of the one we are currently constructing compute the same function as the subformulae which they will represent, then our circuit as a whole computes $f$.

Let $\varphi \in \FO(O, \arb)$ over the signature $\tau$.
For any length $n$ of valid encoded $K$-interpretations $\pi$ for $\tau$, we are going to define a circuit $C_n$ such that $f_{C_n}(\enc(\pi)) = \FOK(O, \arb)\EVAL_\varphi(\pi)$.
If $\varphi$ contains $k$ variables $x_1, \dots, x_k$, we will do this by for each subformula $\psi$ of $\varphi$ and any $(m_1, \dots, m_k) \in A^k$ creating a circuit $C_n^{\psi(m_1, \dots, m_k)}$ such that for any $K$-interpretation $\pi$, where $\lvert \enc(\pi) \rvert = n$, it holds that $\evaluate{\psi[m_1/x_1, \dots, m_k/x_k]}{\pi} = f_{C_n^{\psi(m_1, \dots, m_k)}}(\enc(\pi))$, where $\psi[m_1/x_1, \dots, m_k/x_k]$ is the formula $\psi$, where each occurrence of $x_i$ is replaced by $m_i$ for all $1 \leq i \leq k$.
The $m_1, \dots, m_k$ are essentially used to keep track of the values ``selected'' by the quantifiers and we initialize them to be $0$.

We proceed by structural induction on $\varphi$.
W.l.o.g. let $\varphi$ have exactly $k$ variables.

At the very top is the output gate, so $C_n^{\varphi(0, \dots, 0)}$ consists of the output gate in addition to the circuit as described by the following induction.

For any subformula $\psi$ of $\varphi$, we proceed as follows.
\begin{enumerate}
    \item Let $\psi = \exists x_i \xi$.
    Then $C_n^{\psi(m_1, \dots, m_k)}$ consists of an addition gate with the predecessors $C_n^{\xi(m_1, \dots, m_{i-1}, a, m_{i+1}, \dots,  m_k)}$ for all $a \in A$.
    \item Let $\psi = \forall x_i \xi$.
    Then $C_n^{\psi(m_1, \dots, m_k)}$ is defined as above, except that it has a multiplication gate at the top.
    \item Let $\psi = \xi_1 \lor \xi_2$.
    Then $C_n^{\psi(m_1, \dots, m_k)}$ consists of a sum gate with the predecessors $C_n^{\xi_1(m_1, \dots, m_k)}$ and $C_n^{\xi_2(m_1, \dots, m_k)}$.
    \item Let $\psi = \xi_1 \land \xi_2$.
    Then $C_n^{\psi(m_1, \dots, m_k)}$ is defined as above, except that it has a multiplication gate at the top.
    \item Let $\psi = \xi_i \circ \xi_j$ for $\circ \in O$.
    Then $C_n^{\psi(m_1, \dots, m_k)}$ is defined as above, except that it has a $\circ$ gate at the top.
    \item Let $\psi = x_i \star x_j$ for $\star \in \{=, \neq\}$ and variables $x_i, x_j$.
    Then $C_n^{\psi(m_1, \dots, m_k)}$ consists of a constant $1$ gate if $m_i = m_j$ and a constant $0$ gate, otherwise.
    \item Let $\psi = R(\ol{y})$ for $R \in \tau$.
    Then $C_n^{\psi(m_1, \dots, m_k)}$ is the input gate representing the literal $R(\ol{y})[m_1/x_1, \dots, m_k/x_k]$ in $\enc(\pi)$.
    \item Let $\psi = \neg R(\ol{y})$ for $R \in \tau$.
    Then $C_n^{\psi(m_1, \dots, m_k)}$ is the input gate representing the literal $\neg R(\ol{y})[m_1/x_1, \dots, m_k/x_k]$ in $\enc(\pi)$.    
    \item Let $\psi = P_f(\ol{y})$ for $f \in \arb$.
    Then $C_n^{\psi(m_1, \dots, m_k)}$ is a constant gate that has the value $f(\ol{y}[m_1/x_1, \dots, m_k/x_k])$, where $\ol{y}[m_1/x_1, \dots, m_k/x_k]$ is the tuple of variables $\ol{y}$, where $x_i$ is replaced by $m_i$ for $1 \leq i \leq k$.
\end{enumerate}

This construction ensures that the function defined by each subformula of $\varphi$ is exactly the one of the respective subcircuit and thus the circuit $C_n^{\varphi(0, \dots, 0)}$ computes exactly the function $\evaluate{\varphi}{\pi}$.
Therefore, for each $\FO(O, \arb)$-sentence $\phi$, we have that $\FOK(O, \arb)\EVAL_\phi \in \FACO(O)$ and thus $\FOK(O, \arb) \subseteq \FACO(O)$.

$\FACO[O] \subseteq \FOK(O, \arb)$:

Given a $\FACO(O)$ family $\calC = (C_n)_{n \in \N}$, the idea is to define a single sentence $\varphi$, such that $\evaluate{\varphi}{\pi} = f_\calC(\enc(\pi))$.
The sentence $\varphi$ essentially describes how circuits in the family $\calC$ are evaluated.
It does this by making use of the $\arb$ extension to $\FO$. 
Since we have access to arbitrary functions that may depend on the size of the input $K$-interpretation, the interpretation of these functions can be chosen according to the number of input gates of the circuit. 
This way, $\varphi$ will describe the entire circuit family.
We will use these built-in functions to describe the gates of our circuit. 
In particular, they will give us information about gate types, edges, constant values and indices of input gates.

Let $C_n \in \calC$ and let $q \in \N$ such that $\size(C_n) \leq n^q$.
As per Lemma~\ref{lem:circ_trees}, we can assume that each gate in $C_n$ has fan-out $1$ and that for each gate $g$, each input-$g$ path has the same length.
This essentially gives all of the circuits of $\calC$ a layered form, such that one can talk in an unambiguous way about the depth of any individual gate, in the sense that it is the distance to an input gate.

Additionally, the fact that $\size(C_n) \leq n^q$ allows us to encode each gate in $C_n$ as a $q$ long tuple of values in $\{1, \dots, n\}$.
This will enable us to effectively talk about the structure of $C_n$ logically.

The sentence $\varphi$ will have the signature $\{R^1\}$. 
From the $\arb$ extension, for all $n \in \N$, we also make use of the functions $t_n \colon \{1, \dots, n\}^q \to \{0, 1\}^4$, $c_n \colon \{1, \dots, n\}^q \to K$, $\textit{in}_n \colon \{1, \dots, n\}^{q+1} \to \{0, 1\}$, $e_n \colon \{1, \dots, n\}^{2 \cdot q} \to \{0, 1\}$ and $\textit{left} \colon \{1, \dots, n\}^{2 \cdot q} \to \{0, 1\}$, where $t_n(x_1, \dots, x_q)$ yields the gate type in binary of the gate encoded by $(x_1, \dots, x_q)$ as per Table~\ref{table:gate_types} on page~\pageref{table:gate_types}, $c_n(x_1, \dots, x_q)$ returns the value of the gate encoded by $(x_1, \dots, x_q)$ if it is a constant gate and $0$, if it is not, $\textit{in}_n(x_1, \dots, x_q, y)$ yields $1$, if the gate encoded by $(x_1, \dots, x_q)$ is the $y$th input gate and $0$, if it is not and $e_n(x_1, \dots, x_q, y_1, \dots, y_q)$ returns $1$, if there is an edge from the gate encoded by $(x_1, \dots, x_q)$ to the one encoded by $(y_1, \dots, y_q)$ and $0$, if there is not.
The additional function $\textit{left}$ is only needed if $\leq \in O$ or $< \in O$, and $\textit{left}(x_1, \dots, x_q, y_1, \dots, y_q)$ returns $1$, if the gate encoded by $(x_1, \dots, x_q)$ is the left predecessor of the gate encoded by $(y_1, \dots, y_q)$ and the gate encoded by $(y_1, \dots, y_q)$ is a $<$ gate.
Otherwise, $\textit{left}$ returns $0$.
It is worth to note that $R$, as the unary only relation symbol in $\tau$, yields the individual elements of the input to the circuit.
We make use of that fact when we characterize the input gates logically.

With all that at hand, we will now define $\varphi$ by induction on the layers of the circuit, i.e., we start at depth $0$ and move towards the output gate.
We will do this by defining $\varphi_d$ for each $0 \leq d \leq \depth(C_n)$.

At depth $0$, we have only input gates. 
Therefore $\varphi_0(x_1, \dots, x_q) \coloneqq \\ \exists y~ P_\textit{in}(x_1, \dots, x_q, y) \cdot R(y)$.

For $1 \leq d \leq \depth(C_n)$, we define $\varphi_d$ as follows:

\begin{align*}
    \varphi_d(x_1, \dots, x_q) \coloneqq~& P_t(x_1, \dots, x_q) = (0010) \cdot T_{2, d}(x_1, \dots, x_q) + \\
    & P_t(x_1, \dots, x_q) = (0011) \cdot T_{3, d}(x_1, \dots, x_q) + \\
    & P_t(x_1, \dots, x_q) = (0100) \cdot T_{4, d}(x_1, \dots, x_q) + \\
    & P_t(x_1, \dots, x_q) = (0101) \cdot T_{5, d}(x_1, \dots, x_q) + \\
    & P_t(x_1, \dots, x_q) = (0110) \cdot T_{6, d}(x_1, \dots, x_q) + \\
    & P_t(x_1, \dots, x_q) = (0111) \cdot T_{7, d}(x_1, \dots, x_q) + \\
    & P_t(x_1, \dots, x_q) = (1000) \cdot T_{8, d}(x_1, \dots, x_q) + \\
    & P_t(x_1, \dots, x_q) = (1001) \cdot T_{9, d}(x_1, \dots, x_q),
\end{align*}
where 
\begin{align*}
    T_{2, d}(x_1, \dots, x_q) =~& P_c(x_1, \dots, x_q) \\
    T_{3, d}(x_1, \dots, x_q) =~& \exists y_1, \dots, y_q~ P_e(y_1, \dots, y_q, x_1, \dots, x_q) \land \varphi_{d-1}(y_1, \dots, y_q) \\
    T_{4, d}(x_1, \dots, x_q) =~& \forall y_1, \dots, y_q~ P_e(y_1, \dots, y_q, x_1, \dots, x_q) \land \varphi_{d-1}(y_1, \dots, y_q) \\
    T_{5, d}(x_1, \dots, x_q) =~& \exists y_1, \dots, y_q~ P_e(y_1, \dots, y_q, x_1, \dots, x_q) \land \varphi_{d-1}(y_1, \dots, y_q) \\
    T_{6, d}(x_1, \dots, x_q) =~& \exists y_1, \dots, y_q, z_1, \dots, z_q~ \\
    & P_e(y_1, \dots, y_q, x_1, \dots, x_q) \land P_e(z_1, \dots, z_q, x_1, \dots, x_q) \land \\
    & \left(\bigvee_{1\leq i \leq q} y_i \neq z_i \right) \land \varphi_{d-1}(y_1, \dots, y_q) = \varphi_{d-1}(z_1, \dots, z_q) \\
    T_{7, d}(x_1, \dots, x_q) =~& \exists y_1, \dots, y_q, z_1, \dots, z_q~ \\
    & P_e(y_1, \dots, y_q, x_1, \dots, x_q) \land P_e(z_1, \dots, z_q, x_1, \dots, x_q) \land \\
    & \left(\bigvee_{1\leq i \leq q} y_i \neq z_i \right) \land \varphi_{d-1}(y_1, \dots, y_q) \neq \varphi_{d-1}(z_1, \dots, z_q) \\
    T_{8, d}(x_1, \dots, x_q) =~& \exists y_1, \dots, y_q, z_1, \dots, z_q~ \\
    & P_e(y_1, \dots, y_q, x_1, \dots, x_q) \land P_e(z_1, \dots, z_q, x_1, \dots, x_q) \land \\
    & P_\textit{left}(y_1, \dots, y_q, x_1, \dots, x_q) \land \varphi_{d-1}(y_1, \dots, y_q) \leq \varphi_{d-1}(z_1, \dots, z_q) \\
    T_{9, d}(x_1, \dots, x_q) =~& \exists y_1, \dots, y_q, z_1, \dots, z_q~ \\
    & P_e(y_1, \dots, y_q, x_1, \dots, x_q) \land P_e(z_1, \dots, z_q, x_1, \dots, x_q) \land \\
    & P_\textit{left}(y_1, \dots, y_q, x_1, \dots, x_q) \land \varphi_{d-1}(y_1, \dots, y_q) \not\leq \varphi_{d-1}(z_1, \dots, z_q). 
\end{align*}
Now, it holds for the formula 
\[
    \varphi \coloneqq \exists x_1, \dots, x_q~ P_t(x_1, \dots, x_q) = (0110) \land \varphi_{\depth(C_n)}(x_1, \dots, x_q)
\]
that $\evaluate{\varphi}{\pi} = f_{C_n}(\enc(\pi))$.
Thus, for each $\FACO(O)$-circuit family, there exists a sentence $\varphi$ such that $\FOK(O, \arb)\EVAL_\varphi = f_{C_n}$.
Therefore, $\FOK(O,\arb) \subseteq \FACO(O)$ and putting it all together $\FOK(O,\arb) = \FACO(O)$.
\qed
\end{proof}

\section{Proof of Theorem~\ref{thm:eval-pspace-algo}}\label{prf:eval-pspace-algo}
\begin{restate}[\ref{thm:eval-pspace-algo}]
\begin{theorem}
    Fix a semiring $K$. 
    Given an $\FOK$ formula $\varphi$, a $K$ interpretation $\pi$, and an assignment $s$. 
    The value $\evaluate{\varphi}{\pi,s}$ can be computed in time $\bO(n^2 \cdot |\varphi| \cdot |\pi|^{|\varphi|})$ and space in $\bO(\mathrm{poly}(n))$, with $n = |\varphi|+|\pi|+|s|$ the sum of the encoding lengths.
\end{theorem}
\end{restate}
\begin{algorithm}[t]
    \DontPrintSemicolon
    \SetKwInOut{Data}{Data}
    \SetKwProg{Proc}{Procedure}{}{}
    \KwData{set $A$, $K$-interpretation $\pi$}
    \Proc{\procEval(formula $\varphi$, assignment $s$)}{
    \Switch(\tcp*[f]{time complexity}){$\varphi$}{
        \lCase(\tcp*[f]{$\bO(n)$}){$x = y$}{\Return 1 iff $s(x) = s(y)$}
        \lCase(\tcp*[f]{$\bO(n)$}){$x \neq y$}{\Return 1 iff $s(x) \neq s(y)$}
        \lCase(\tcp*[f]{$\bO(n^2)$}){$ R(\bar x)$}{\Return $\pi(R(s(\bar x)))$}
        \lCase(\tcp*[f]{$\bO(n^2)$}){$\lnot R(\bar x)$}{\Return $\pi(\lnot R(s(\bar x)))$}
        \lCase(\tcp*[f]{$\bO(n) + t(\psi) + t(\theta)$}){$\psi \land \theta$}{\Return $\procEval(\psi, s) \cdot \procEval(\theta, s)$}
        \lCase(\tcp*[f]{$\bO(n) + t(\psi) + t(\theta)$}){$\psi \lor \theta$}{\Return $\procEval(\psi, s) + \procEval(\theta, s)$}
        \lCase(\tcp*[f]{$\bO(n) + t(\psi) + t(\theta)$}){$\psi \star \theta$}{\Return 1 iff $\procEval(\psi, s) \star \procEval(\theta, s)$}
        \lCase(\tcp*[f]{$|A| \cdot \bO(n)\cdot t(\psi)$}){$\exists x \psi$}{\Return $\sum_{a \in A}\procEval(\psi, s[a/x])$}
        \lCase(\tcp*[f]{$|A| \cdot \bO(n)\cdot t(\psi)$}){$\forall x \psi$}{\Return $\prod_{a \in A}\procEval(\psi, s[a/x])$} 
    }
    }
\caption{Evaluation of $\evaluate{\varphi}{\pi, s}$, where $\star\in O\subseteq\{=,\neq,\leq,\not\leq\}$ }
\label{alg:evaluation of varphi}
\end{algorithm}
\begin{proof}
The procedure $\procEval$ in Algorithm~\ref{alg:evaluation of varphi} is a recursive algorithm that runs on a $K$-TM to solve the model checking problem. 
Set $A$ and $\pi$ are used as ``global variables'' as they are never modified in the recursive steps. 
They are accessible by every recursive algorithmic call. 
The correctness follows inductively by semantics (Def.~\ref{interpretation}). 

Now let $n$ be the input length, i.e., $n=|\varphi|+|\pi|+|s|$. 
Notice that the universe $A$ can be explicitly extracted from $\pi$ in time $O(|\pi|^2)$. 
We can see this as a preprocessing step before the execution of Algorithm~\ref{alg:evaluation of varphi}. 
As a result, in the following, we will use $A$ as if it is an input. 

To measure the space used by the machine, we need to prove an upper bound on the recursion depth and the space used in a recursive step. 
For every conjunction and disjunction there are two recursive steps. 
For every quantifier there are $|A|$-many recursive steps. 
So altogether, we can bound the number of recursive steps as follows for a given formula $\varphi$:
\[
    (2\cdot(\#_\land(\varphi)+\#_\lor(\varphi))+1) \cdot |A|^{(\#_\exists(\varphi)+\#_\forall(\varphi))} \in \bO(|\varphi|\cdot|A|^{|\varphi|}),
\]
where $\#_O(\varphi)$, for $O\in\{\land,\lor,\exists,\forall\}$, is the number of occurrences of $O$ in $\varphi$. 
Now, we turn towards the space and time bound of a single recursive step. 
We do a case distinction according to the \textbf{switch}-expression in the algorithm.
\begin{description}
    \item[$x=y$ / $x\neq y$:] Use separate registers in $R$ to copy and check if such an expression is true. This needs constant space and linear time in $n$.
    \item[$R(\bar x)$ / $\lnot R(\bar x)$:] Copy the values of $s(\bar x)$ to the end of the tape and return the specified value according to $\pi$. For that purpose, we need additional markings to ``remember'' which positions have been compared. Altogether this can be done in quadratic time in $n$ and linear space in $n$.
    \item[$\wedge$ / $\vee$:] Here, we need to copy the respective parts from the input yielding $\bO(n)$ time and space.
    \item[$\star$:] Analogously as for the previous case.
    \item[$\exists$ / $\forall$:] 
    Again, we essentially need to copy parts from the input and patch the assignment. 
    Regarding the $A$-values we need to iterate through this part of the input yielding $\bO(n)$ time and space.
\end{description}
We see that the space of each step is bounded by $\bO(n)$. 
Regarding time complexity, 
the time needed at each step is in $\bO(n^2)$ and the number of recursive steps was bounded by $\bO(|\varphi| \cdot |A|^{\varphi})$
yielding a time bound of $\bO(n^2 \cdot |\varphi| \cdot |A|^{|\varphi|})$.\qed
\end{proof}
\end{document}